\documentclass[letterpaper, 10 pt, journal, twoside]{IEEEtran}
\IEEEoverridecommandlockouts

\usepackage{cite}
\usepackage{amsmath,amssymb,amsfonts}
\usepackage{bm}
\usepackage{algorithmic}
\usepackage{graphicx}
\usepackage{textcomp}
\usepackage{xcolor,color}
\usepackage{amsthm}
\usepackage[ruled,vlined]{algorithm2e}
\usepackage[top=0.833in, left=0.667in, right=0.667in, bottom=0.597in, includefoot]{geometry}
\usepackage{hyperref}
\usepackage{tikz}
\usetikzlibrary{positioning,automata}
    
\makeatletter
\let\MYcaption\@makecaption
\makeatother 

\usepackage[font=footnotesize]{subcaption}

\makeatletter
\let\@makecaption\MYcaption
\makeatother

\newcommand{\until}[2]{\cup_{[#1,#2]}}
\newcommand{\always}[2]{\square_{[#1,#2]}}
\newcommand{\eventually}[2]{\lozenge_{[#1,#2]}}

\newcommand{\x}{\bm{x}}
\newcommand{\z}{\bm{z}}
\newcommand{\A}{\bm{A}}
\renewcommand{\c}{\bm{c}}
\newcommand{\s}{\mathbf{s}}

\newtheorem{definition}{Definition}

\newtheorem{theorem}{Theorem}

\newtheorem{remark}{Remark}
\newtheorem{problem}{Problem}
\newtheorem{assumption}{Assumption}

\newcommand{\hl}[1]{#1}

\begin{document}

\title{A More Scalable Mixed-Integer Encoding for Metric Temporal Logic}

\author{Vince Kurtz and Hai Lin
\thanks{The authors are with the Departments of Electrical Engineering, University of Notre Dame, Notre Dame, IN, 46556 USA. \texttt{\{vkurtz,hlin1\}@nd.edu}}
\thanks{This work was supported by NSF Grants CNS-1830335, IIS-2007949.}
}

\maketitle
\thispagestyle{empty}

\begin{abstract}
    The state-of-the-art in optimal control from timed temporal logic specifications, including Metric Temporal Logic (MTL) and Signal Temporal Logic (STL), is based on Mixed-Integer Convex Programming (MICP). The standard MICP approach is sound and complete, but struggles to scale to long and complex specifications. Drawing on recent advances in trajectory optimization for piecewise-affine systems, we propose a new MICP encoding \hl{for finite transition systems} that significantly improves scalability to long and complex MTL specifications. Rather than seeking to reduce the number of variables in the MICP, we focus instead on designing an encoding with a tight convex relaxation. This leads to a larger optimization problem, but significantly improves branch-and-bound solver performance. In simulation experiments involving a mobile robot in a grid-world, the proposed encoding can reduce computation times by several orders of magnitude. 
\end{abstract}

\section{Introduction and Related Work}\label{sec:intro}

Timed temporal logics like Metric Temporal Logic (MTL) and Signal Temporal Logic (STL) offer a compact means of expressing complex specifications with timing constraints. Efficient methods of enabling a system to satisfy logical specifications are of particular interest in robotics and Cyber-Physical Systems. For example, a mobile robot need to visit several types of waypoints before a deadline (see Fig.~\ref{fig:multitarget}). 

Early work on the synthesis problem focused primarily on automata-based methods \cite{baier2008principles,alur2015principles}, which face severe scalability challenges. A prominent alternative is to encode the problem as a satisfiability problem (SAT/SMT) \cite{cimatti2002nusmv, de2008z3}. These methods tend to scale well \cite{shoukry2016scalable}, but solutions are not  globally optimal, even with convex optimization as a theory solver \cite{shoukry2017smc}. 

In this paper, we focus on the Mixed-Integer Convex Programming (MICP) approach to synthesis \cite{raman2014model,karaman2008vehicle}. In addition to being sound (any solution satisfies the specification) and complete (a solution will be found if one exists), MICP methods are guaranteed to find a globally optimal solution. Furthermore, the MICP paradigm generalizes naturally to systems with high-dimensional dynamics \cite{belta2019formal} and allows for maximizing the STL robustness measure \cite{sadraddini2015robust}. Existing MICP methods tend to be slower than SAT/SMT methods \cite{shoukry2016scalable}, however. 

\begin{figure}
    \centering
    \includegraphics[width=0.8\linewidth]{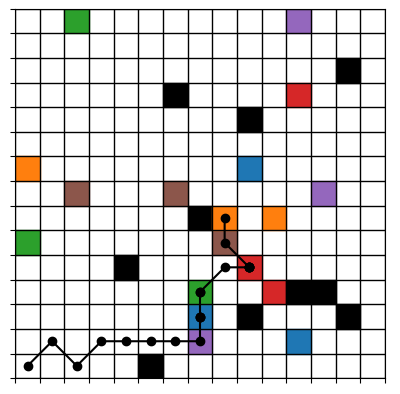}
    \caption{A robot in a grid-world is tasked with visiting six types of waypoints (red, green, blue, yellow, purple and brown squares) while avoiding obstacles (black squares). A standard mixed-integer encoding takes \textbf{1 hour 25 minutes} to find an optimal path, while our proposed approach takes only \textbf{26 seconds}.}
    \label{fig:multitarget}
\end{figure}

The primary drawback of the MICP approach is scalability \cite{belta2019formal}. Standard MICP encodings introduce a new binary variable for each timestep and each sub-formula in the specification. Since the worst-case complexity of MICP is exponential in the number of binary variables, performance rapidly degrades for complex specifications and long time-horizons. 

For this reason, much research has focused on encodings with fewer binary variables and constraints. For example, \cite{sadraddini2018formal} show that half as many constraints can be used if specifications are first re-written in Postive Normal Form (PNF), and \cite{saha2016milp} iteratively solve a sequence of smaller MICPs. 

More recently, there has been a trend toward avoiding integer programming entirely. For certain fragments of STL and MTL \hl{(related logics that exclude certain operators or combinations of operators)}, synthesis can be done with convex programming directly \cite{raman2014model}. Similar fragments have been \hl{used} for synthesis based on Control Barrier Functions (CBFs) \cite{lindemann2018control} and learning \cite{aksaray2016q}. For STL, smooth approximations of the robustness score \cite{pant2017smooth,pant2018fly,mehdipour2019arithmetic,gilpin2021smooth} have been used to find local solutions via gradient descent.

None of the above methods are sound and complete for the full syntax of MTL or STL, however. This limits them to relatively simple scenarios, and they may struggle to find a solution (or be heavily dependent on an initial guess) in the case of more complex specifications like the one in Fig.~\ref{fig:multitarget}. 

In this work, we revisit MICP for timed temporal-logic, focusing in particular on MTL. Taking inspiration from recent work on trajectory-optimization for piecewise-affine systems \cite{marcucci2021shortest}, we note that while the worst-case complexity of MICP depends primarily on the number of binary variables, performance in practice often depends more heavily on the tightness of the convex relaxation \hl{(a convex program where the binary variables in the original MICP are allowed to take continuous values)}, since modern MICP solvers rely heavily on the branch-and-bound algorithm \cite{gurobi}. 

Following \cite{marcucci2021shortest}, the basic idea behind our proposed encoding is to introduce binary variables for each possible transition at every timestep, rather than for each possible state. This results in more binary variables than a standard encoding, but a tighter convex relaxation and better performance in practice \hl{\cite{sherali1990hierarchy}}. This approach holds notable similarities to mixed-integer encodings for various problems in graph theory, including the shortest path problem (SPP) \cite{ahuja1993network} and traveling salesman problem (TSP) \cite{miller1960integer}, which can be viewed as special cases of MTL. 

Our primary contributions are summarized as follows:
\begin{enumerate}
    \item We present a more scalable mixed-integer encoding for \hl{finite transition systems subject to} MTL specifications. 
    \item Our proposed method is sound and complete.
    \item The convex relaxation of our proposed encoding is at least as tight as that of a standard MICP encoding.
    \item The convex fragment induced by our proposed encoding is strictly larger than that of standard MICP encoding. 
    \item In simulation experiments, our proposed approach outperforms standard MICP and SAT-based synthesis. 
\end{enumerate}

The remainder of this paper is organized as follows: background and a formal problem statement is presented in Section~\ref{sec:background}. We summarize the standard MICP approach in Section~\ref{sec:standard} and present our proposed encoding in Section~\ref{sec:main_results}, along with proofs of soundness, completeness, and convex relaxation tightness. We provide simulation examples in Section~\ref{sec:simulation} and conclude with Section~\ref{sec:conclusion}.

\section{Background}\label{sec:background}

\subsection{Metric Temporal Logic}\label{subsec:mtl}

Metric Temporal Logic (MTL) is an extension of Linear Temporal Logic (LTL) which allows for timing-related deadlines \cite{koymans1990specifying}. The syntax of MTL is defined as:
\begin{equation}\label{eq:mtl_syntax}
   \varphi := \top \mid \pi \mid \lnot \varphi \mid \varphi_1 \land \varphi_2 \mid \varphi_1 \until{t_1}{t_2} \varphi_2 
\end{equation}
where $\pi \in AP$ is an atomic proposition, boolean operators ``not'' ($\lnot$) and ``and'' ($\land$) can be used to define disjunction ($\lor$), and the temporal operator ``until'' ($\until{t_1}{t_2}$) can be used to define ``always'' ($\always{t_1}{t_2}$) and ``eventually'' ($\eventually{t_1}{t_2}$). We assume bounded-time specifications, i.e., $t_2$ is finite:

\begin{assumption}[Bounded-time Specification]\label{assumption:bounded_mtl}
    There exists some $0 \leq T < \infty$ such that the satisfiability of $\varphi$ can be uniquely determined in $T$ timesteps.
\end{assumption}

MTL semantics are defined over words $\bm{\sigma} = \sigma_0,\sigma_1,\sigma_2,\dots$, where $\sigma_t \in 2^{AP}$ is the set of atomic propositions that hold at timestep $t$. We denote that $\bm{\sigma}$ satisfies the MTL formula $\varphi$ with $\bm{\sigma} \vDash \varphi$, and that the suffix $\sigma_t,\sigma_{t+1},\dots$ satisfies $\varphi$ with $\bm{\sigma} \vDash_t \varphi$. MTL semantics are defined recursively as follows:
\begin{alignat*}{3}
    & \bm{\sigma} \vDash_t \pi                       && \iff \pi \in \sigma_t \\
    & \bm{\sigma} \vDash_t \lnot \varphi             && \iff \bm{\sigma} \nvDash_t \varphi \\
    & \bm{\sigma} \vDash_t \varphi_1 \land \varphi_2 && \iff \bm{\sigma} \vDash_t \varphi_1 \text{ and } \bm{\sigma} \vDash_t \varphi_2 \\
    & \bm{\sigma} \vDash_t \varphi_1\until{t_1}{t_2}\varphi_2 && \iff \exists t' \in [t+t_1,t+t_2] \text{ s.t. } \bm{\sigma} \vDash_{t'} \varphi_2 \\ 
    & && \qquad\quad \text{ and } \forall t'' \in [t,t'-1], \quad \bm{\sigma} \vDash_{t''} \varphi_1\\
    & \bm{\sigma} \vDash \varphi                   && \iff \bm{\sigma} \vDash_0 \varphi 
\end{alignat*}

\begin{remark}
    MTL is closely related to Signal Temporal Logic (STL), though MTL is defined over discrete atomic propositions, whereas STL is defined over continuous-valued signals. 
\end{remark}

\subsection{Mixed Integer Programming}\label{subsec:mip}

MICP considers problems of the form 
\begin{subequations}\label{eq:general_micp}
\begin{align}
   \min_{\x,\z} ~& J(\x,\z) \\
   \text{s.t. } & \A \begin{bmatrix} \x \\ \z \end{bmatrix} \leq \c
\end{align}
\end{subequations}
where $\x$ is a vector of real-valued decision variables, $\z$ is a vector of binary-valued%
\footnote{Note that in general $\z$ can take integer values. In most applications of MICP, however, only binary variables are considered. For this reason, we will refer to $\z$ as binary variables throughout this paper.}
decision variables, $\A$ and $\c$ are matrices of appropriate dimensions, and $J(\cdot,\cdot)$ is convex. 

If the values of $\z$ are fixed, (\ref{eq:general_micp}) can be solved rapidly with convex programming. But trying every possible $\z$ would be prohibitively expensive. Fortunately, there are several tricks that allow MICP solvers to mostly avoid this worst-case scenario \hl{\cite[Chapter 1.2]{conforti2014integer}}. The most prominent such method is \textit{branch-and-bound}.

The main idea behind branch-and-bound is to ``branch'' on some of the binary variables by fixing their values as 0 or 1. With these variables fixed, we have a smaller MICP with fewer binary variables. We then obtain a ``bound'' by solving a convex relaxation: all of the other binary variables are allowed to take continuous values in $[0,1]$. This convex relaxation can be solved quickly with specialized convex optimization methods, and provides a lower bound on the optimal cost. 

If the convex relaxation is infeasible, we know that the binary values that we fixed are incorrect, allowing us to rapidly eliminate all of the possible solutions with those values. Furthermore, any integer-feasible solution provides an upper bound on the optimal cost. If a given branch has a convex relaxation with a higher cost than this upper bound, we can similarly eliminate solutions in that branch. 

Clearly, the tightness of the convex relaxation has a significant impact on the efficiency of the branch-and-bound algorithm. A tighter convex relaxation will allow the solver to eliminate suboptimal branches rapidly, while a loose convex relaxation is more likely to result in the worst-case scenario of fully exploring every possible branch. With this in mind, we define the \textit{relaxation gap}:

\begin{definition}[Relaxation Gap]
    Given a MICP of the form (\ref{eq:general_micp}), let $J^*$ be the optimal cost and $\tilde{J}^*$ be the optimal cost of the convex relaxation. Then the relaxation gap is given by 
    \begin{equation}\label{eq:relaxation_gap}
        r_{gap} = (J^* - \tilde{J}^*)/J^* .
    \end{equation}
\end{definition}

If the relaxation gap is zero, integer constraints are unnecessary and the problem can solved with convex programming directly. This is the case for a number of interesting problems, including the SPP in graph theory \cite{ahuja1993network} and a convex fragment of timed temporal logic \cite[Theorem 1]{raman2014model}. Even for nonconvex problems (e.g., MTL and the TSP), the relaxation gap plays an important role in determining solver performance \hl{\cite{marcucci2021shortest, sherali1990hierarchy}}.

\subsection{Problem Formulation}\label{subsec:problem_formulation}

In this paper, we consider synthesis over finite-state labeled transition systems:

\begin{definition}[Transition System]
    A transition system 
    \begin{equation*}
        TS = (S, s_0, \to, AP, L, C)
    \end{equation*}
    is a tuple consisting of the following elements:
    \begin{itemize}
        \item $S$ is a finite set of states
        \item $s_0 \in S$ is an initial state
        \item $\to \subseteq S \times S$ are transition relations
        \item $AP$ is a finite set of atomic propositions
        \item $L : S \to 2^{AP}$ is a labeling function
        \item $C : (S \times S) \to \mathbb{R}^+$ is a cost function
    \end{itemize}
\end{definition}

%
%
Note that each transition in $TS$ is associated with a cost. Our goal will be to find a minimum-cost path through $TS$ that satisfies a given MTL specification.

We now provide several definitions to allow for more efficient discussion of transition systems. First, a \textit{path} is merely a sequence of states that obeys the transition relations:
\begin{definition}[Path]
    A sequence of states 
    $
        \s = s_0,s_1,\dots,s_T
    $
    is a path of TS if $(s_t,s_{t+1}) \in \to ~~ \forall t \in [0,\hl{T-1}]$.
\end{definition}
Note that we focus in this paper on finite-length paths. Every path is associated with a total cost:
\begin{definition}[Path Cost]
    Given path $\s$,
    $
        J(\s) = \sum_{t=0}^{T-1} C(s_t,s_{t+1})
    $
    is the total path cost associated with $\s$.
\end{definition}
In addition to the cost of any given path, we are also interested in the corresponding sequence of atomic propositions:
\begin{definition}[Trace]
   The trace of path $\s$ is given by
    $
       \bm{\sigma}(\s) = L(s_0),L(s_1),\dots,L(s_T).
    $
\end{definition}

Note that the trace is a word over which we can consider satisfaction of an MTL formula. With some liberty of notation, we write $\s \vDash \varphi$ if $\bm{\sigma}(\s) \vDash \varphi$. Finally, we define the set of adjacent states as those that can be transitioned to from the current state:
\begin{definition}[Adjacent Set]\label{def:adj}
    Given state $s \in S$, the adjacent set of $s$ is given by
    $
        Adj(s) = \{ s' \in S \mid (s,s') \in \to \}.
    $
\end{definition}

\hl{ Note that if $TS$ contains self-loops, i.e., $(s_i,s_i) \in \rightarrow$, then $s_i \in Adj(s_i)$.} We can now provide a formal problem statement:

\begin{problem}\label{prob:main_problem}
    Given transition system $TS$ and bounded-time MTL specification $\varphi$, find the minimum-cost path through $TS$ that satisfies $\varphi$, i.e.,
    \begin{subequations}\label{eq:main_problem}
    \begin{align}
        \min_{\s} ~& J(\s) \label{eq:main_problem_cost}\\
        \hl{\mathrm{s.t.~}} & \s \vDash \varphi. \label{eq:main_problem_constraint}
    \end{align}
    \end{subequations}
\end{problem}

\section{Standard Mixed-Integer Encoding}\label{sec:standard}

In this section, we present the standard method of encoding (\ref{eq:main_problem}) as an MICP. Our presentation in this section is based primarily on \cite{raman2014model} and \cite{belta2019formal}, which consider STL specifications, but similar encodings for MTL are also popular \cite{karaman2008vehicle,saha2016milp}.

The basic idea is to introduce a binary variable $b_s(t)$ for each state $s$ and each timestep $t$\footnote{\hl{In the STL case, binary variables are introduced for each predicate and subformula, which correspond to states if the system is abstracted as a finite-state transition system.}}. We pose the problem such that $b_s(t) = 1$ means the optimal path visits state $s$ at time $t$. 

We start with the following dynamics constraints:
\hl{
\begin{subequations}\label{eq:transition_system_constraints}
\begin{gather}
    b_{s_0}(0) = 1, \label{eq:standard_initial_constraint} \\
    \sum_{s \in S} b_s(t) = 1 \quad \forall t \in [0,T], \label{eq:standard_occupation_constraint}\\
    \sum_{s' \in Adj(s)} b_{s'}(t+1) \geq b_s(t) \quad \forall s \in S, \label{eq:standard_transition_constraint}
\end{gather}
\end{subequations}}
where (\ref{eq:standard_initial_constraint}) establishes $s_0$, (\ref{eq:standard_occupation_constraint}) ensures that only one state can be occupied at each timestep, and  (\ref{eq:standard_transition_constraint}) enforces transition relations.

To satisfy the specification, we add additional variables and constraints which are defined recursively. The main insight is to define new binary variables, $z^\varphi(t)$, such that $z^\varphi(t) = 1$ only if $\varphi$ is satisfied starting from time $t$. First, note that conjuction and disjunction can encoded as linear constraints as follows:
\begin{gather}
    z \vDash \bigwedge_{i=1}^n z_i \iff z \leq z_i ~\forall i \text{ and } z \geq \hl{1 - n + \sum_{i=1}^n z_i } \\
    z \vDash \bigvee_{i=1}^n z_i \iff z \leq \sum_{i=1}^n z_i \text{ and } z \geq z_i ~\forall i
\end{gather}

%
This allows us to encode satisfaction of $\varphi$ as follows:
\begin{subequations}\label{eq:mtl_constraints}
\begin{alignat}{2}
    & \pi                       && \implies z^\pi = \sum_{s \in S \mid L(s) = \pi} b_s(t), \label{eq:standard_AP_encoding} \\
    & \lnot \varphi             && \implies z^{\lnot \varphi} = 1 - z^{\varphi}, \\
    & \varphi_1 \land \varphi_2 && \implies z^{\varphi_1 \land \varphi_2} = z^{\varphi_1} \land z^{\varphi_2}, \\
    & \varphi_1\until{t_1}{t_2}\varphi_2 && \implies \\
    & &&\bigvee_{t' \in [t+t_1,t+t_2]} \left(z^{\varphi_2}(t') \land \bigwedge_{t'' \in [t,t'-1]} z^{\varphi_1}(t'') \right). \nonumber
\end{alignat}
\end{subequations}


With this in mind, we can write problem (\ref{eq:main_problem}) as follows:
\begin{subequations}\label{eq:standard_encoding}
\begin{align}
    \min_{b_s, z} ~& J(\hl{b_s}) \\
    \text{s.t. } & \text{Transition System Constraints \hl{(\ref{eq:transition_system_constraints})}}, \label{eq:standard_ts_constraints} \\ 
                 & \text{MTL Constraints (\ref{eq:mtl_constraints})}, \\
                 & z^\varphi(0) = 1. \label{eq:standard_satisfaction_constraint}
\end{align}
\end{subequations}

We focus here on finding a minimum-cost path, but a (convex) cost function can also be designed for different purposes, such as maximizing the STL robustness score \cite{sadraddini2015robust}. 

The scalability limitations of this standard MICP formulation are well-known \cite{belta2019formal,mehdipour2019arithmetic,gilpin2021smooth,pant2017smooth,kurtz2020trajectory}. In particular, this encoding is associated with rapidly increasing solve times in the case of long time horizons (which increase the number of binary variables) and complex specifications (which increase the complexity of the constraint structure). 

\section{Main Results}\label{sec:main_results}

In this section, we exploit the fact that is often the tightness of convex relaxation, rather than \hl{the} number of binary variables and constraints, that determines MICP scalability to propose a more efficient MICP encoding. Our main inspiration in this regard is \cite{marcucci2021shortest}, which presents a more efficient MICP for control of PWA systems by increasing the number of binary variables but tightening the convex relaxation.

We begin by constructing a directed graph $G = (V, E)$, associated with $TS$. Each node $i \in V$ corresponds to a state $s \in S$ and a timestep $t$, i.e., $i = (s,t)$. Each edge $(i,j) \in E$ connects two nodes only if there is a corresponding transition: 
\begin{equation*}
    \big( (s,t), (s', t+1) \big) \in E \iff s' \in Adj(s).
\end{equation*}
Additionally, for each node $i \in V$ we define the input set $I_i = \{j \mid (j,i) \in E \}$ and the output set $O_i = \{j \mid (i,j) \in E \}$. 

Our basic idea is to introduce a binary variable for every edge in the graph (the standard encoding (\ref{eq:standard_encoding}) introduces a binary variable for each node). This may seem counterintuitive, as there are many more edges than nodes, but similar formulations perform well for special cases of temporal-logic planning, including SPP \cite{ahuja1993network}, TSP \cite{miller1960integer}, and PWA control \cite{marcucci2021shortest}.

More specifically, we define binary variables $a_{ij}$ for each edge, where $a_{ij}$ will take unit value only if the edge $(i,j)$ is part of the optimal satisfying path. We can then \hl{implicitly define} variables $b_i = b_s(t)$ representing the total flow through each node as follows\hl{, since $i = (s,t)$}:
\begin{equation}\label{eq:our_occupancy_constraints}
    b_s(t) = 
    \begin{cases}
        \sum_{j \in O_i} a_{ij} & \text{if } t = 0 \\
        \sum_{j \in I_i} a_{ji} & \text{otherwise}.
    \end{cases}
\end{equation}
These flow variables take binary values at optimality, \hl{and} can be used to enforce MTL constraints following (\ref{eq:mtl_constraints}).

Our proposed MICP encoding can then be written as:
\begin{subequations}\label{eq:our_encoding}
\begin{align}
    \min_{a_{ij}, b_s, z} ~& J(\hl{a_{ij}}) \\
    \text{s.t. } & \sum_{j \in O_i} a_{ij} - \sum_{j \in I_i} a_{ji} = 
                    \begin{cases}
                        1 & \text{if } i = (s_0, 0) \\
                        0 & \text{if } t < T
                    \end{cases} \label{eq:our_flow_constraints} \\
                 & \text{Occupancy constraints (\ref{eq:our_occupancy_constraints})} \label{eq:our_MICP_occupancy} \\
                 & \text{MTL Constraints (\ref{eq:mtl_constraints})} \\
                 & z^\varphi(0) = 1 \label{eq:our_satisfaction_constraint}
\end{align}
\end{subequations}
where (\ref{eq:our_flow_constraints}) establishes flow constraints and an initial state, and (\ref{eq:our_MICP_occupancy}-\ref{eq:our_satisfaction_constraint}) enforce satisfaction of the MTL formula. 

The proposed encoding is sound and complete:

\begin{theorem}
    Any solution to (\ref{eq:our_encoding}) satisfies the specification $\varphi$ (soundness); and if a satisfying path exists which satisfies $\varphi$, then a solution to (\ref{eq:our_encoding}) exists (completeness).
\end{theorem}
\begin{proof}
    The theorem follows from the inclusion of constraints (\ref{eq:our_MICP_occupancy}-\ref{eq:our_satisfaction_constraint}) and \cite[Theorem 1]{belta2019formal}.
\end{proof}

Furthermore, the relaxation gap of our proposed encoding is no greater than that of the standard encoding for all MTL specifications:

\begin{theorem}
    Let $r_{gap}^*$ be the relaxation gap associated with (\ref{eq:our_encoding}) and $r_{gap}$ be the relaxation gap associated with (\ref{eq:standard_encoding}). Then
    \begin{equation*}
        r_{gap}^* \leq r_{gap}.
    \end{equation*}
\end{theorem}
\begin{proof}
    First, note that the optimal (non-relaxed) cost for both (\ref{eq:standard_encoding}) and (\ref{eq:our_encoding}) are the same. 
    
    Next, note that any solution of a convex relaxation of (\ref{eq:our_encoding}) is also a valid solution to a convex relaxation of (\ref{eq:standard_encoding}). This is because the (non-binary) flow constraints (\ref{eq:our_flow_constraints}) are sufficient for enforcing the (non-binary) transition constraints (\ref{eq:standard_initial_constraint}-\ref{eq:standard_transition_constraint}). 
    
    Now assume that we have $r_{gap}^* > r_{gap}$. That would mean that the optimal cost associated with a convex relaxation of (\ref{eq:our_encoding}) is less than the optimal cost associated with a convex relaxation of (\ref{eq:standard_encoding}). But this is a contradiction, since the (relaxed) solution to (\ref{eq:our_encoding}) is also a solution to (\ref{eq:standard_encoding}). 
    
    Furthermore, the converse is not always true. Specifically, the transition constraint (\ref{eq:standard_transition_constraint}) allows some ``flow'' to pass between non-adjacent nodes in a convex relaxation. Such non-adjacent flows are not allowed in (\ref{eq:our_encoding}). This means that it is often the case that $r^*_{gap} \ll r_{gap}$ (see Section \ref{sec:simulation}). 
\end{proof}

\begin{remark}
    \hl{While the worst-case complexity of (\ref{eq:our_encoding}) is higher than that of (\ref{eq:standard_encoding}), as there are more binary variables, the relative tightness of the convex relaxation leads to better scalability in practice, as shown in Section \ref{sec:simulation}.}
\end{remark}

Finally, we show that the relaxation gap is zero for a surprisingly large fragment of MTL:

\begin{theorem}
   For MTL specifications $\varphi$ belonging to the fragment
   \begin{equation}\label{eq:convex_fragment}
   \begin{gathered}
        \psi := \pi \mid \psi_1 \land \psi_2 \mid \psi_1 \lor \psi_2 \\
        \varphi := \always{t_1}{t_2}\psi \mid \eventually{t_2}{t_2} \psi \mid \psi_1 \until{t_2}{t_2} \psi_2 \mid \varphi_1 \land \varphi_2
   \end{gathered}
   \end{equation}
   the relaxation gap associated with (\ref{eq:our_encoding}) is zero. 
\end{theorem}

\begin{proof}
    For formulas over which only conjunction is used, the problem is convex and thus the relaxation gap is zero. This follows from \cite[Theorem 1]{raman2014model}. Note that because the interval associated with the temporal operators ``eventually'' and ``until'' is a single timestep, disjunctions in the MTL encoding (\ref{eq:mtl_constraints}) occur only between state formulas $\psi$. 
    
    With this in mind, consider the case where the formula contains at least one set of disjunctions over state formulas, i.e., $\bigvee_{k} \psi_k$ and there are some flows $a_{ij} \in (0,1)$ that take non-binary values at optimality. We will show that in this case, there is always a binary solution that achieves the same cost. 
    
    Note that for the fragment (\ref{eq:convex_fragment}), any $a_{ij} \in (0,1)$ arise only due to several possible ``paths'' satisfying different possible state formulas $\psi_k$. Furthermore, each of these paths must have equal cost. This is easily established by contradiction: if the paths do not have equal cost, a lower-cost solution can be obtained by following only the lower-cost paths. Therefore following any single path ($a_{ij} \in \{0,1\}$) results in the same cost. Thus the relaxation gap is zero and the theorem holds.
\end{proof}

This means that for specifications in the fragment (\ref{eq:convex_fragment}), the synthesis problem can be solved in polynomial time using linear programming. This fragment is significantly more expressive than the convex fragment associated with the standard MICP encoding (\ref{eq:standard_encoding}), which only considers atomic propositions, conjunctions, and the ``always'' operator \cite{raman2014model}. 


\begin{remark}
    The fragment (\ref{eq:convex_fragment}) is the same as that considered in \cite{kurtz2020trajectory}, which presents a scalable but incomplete synthesis method. In contrast, our proposed encoding is sound and complete for all specifications, including those in this fragment.
\end{remark}

It may be somewhat surprising that disjunctions, which would seem to introduce some sort of inherently combinatorial aspect to the problem, can be included in a convex fragment. This sort of convexity despite the presence of disjunctions is a feature shared with the LP encodings of the SPP, where the solver must choose between several seeming disjointed paths, but a convex formulation is possible. 


\section{Simulation Experiments}\label{sec:simulation}

In this section, we demonstrate the scalability of our proposed encoding (\ref{eq:our_encoding}) on several robot motion planning problems. The transition system $TS$ models a robot in an $N\times N$ grid-world, as shown in Fig. \ref{fig:multitarget}. Each state $s \in S$ corresponds to a grid cell. Transitions to adjacent cells (including diagonals) are associated with cost $1$ while transitions to the same cell have cost $0$. No other transitions are allowed.

All experiments were performed on a laptop (i7 processor, 32GB RAM) using Gurobi \cite{gurobi} (version 9.0.3, default options) as the MICP solver. Drake \cite{drake} python bindings were used to interface with the solver.

In addition to comparing our proposed MICP encoding (\ref{eq:our_encoding}) with the standard MICP encoding (\ref{eq:standard_encoding}), we consider an SAT-based approach in which the constraints (\ref{eq:standard_ts_constraints}-\ref{eq:standard_satisfaction_constraint}) are passed to the z3 SAT solver \cite{de2008z3}. This method returns a non-optimal solution and tends to be faster than the standard MICP. 

We first consider the simple reach-avoid scenario shown in Fig.~\ref{fig:reach_avoid}, where a robot must reach a goal (green) and avoid an obstacle (red). The atomic propositions for this scenario are $AP = \{goal, obstacle\}$ and the specification is given by
\begin{equation*}
    \lnot obstacle \until{0}{T} goal,
\end{equation*}
where we chose $T = 15$. 

\begin{figure}
    \centering
    \begin{subfigure}{0.48\linewidth}
        \centering
        \includegraphics[width=0.8\linewidth]{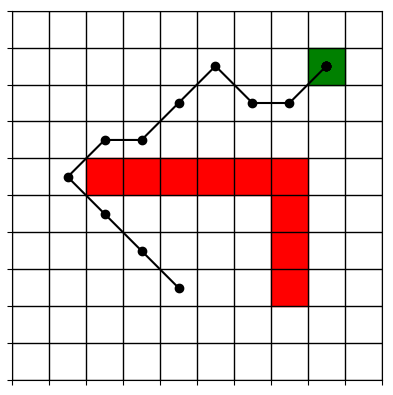}
        \caption{Standard MICP: \textbf{12.7s}. Ours: \textbf{0.47s}.}
        \label{fig:reach_avoid}
    \end{subfigure}
    \begin{subfigure}{0.48\linewidth}
        \centering
        \includegraphics[width=0.8\linewidth]{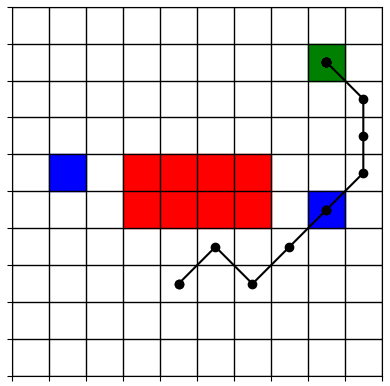}
        \caption{Standard MICP: \textbf{25.9s}. Ours: \textbf{0.61s}.}
        \label{fig:either_or}
    \end{subfigure}
    \caption{Simple robot motion planning specifications in which the robot must navigate through a grid-world and reach a goal (green) and an intermediate target (blue) while avoiding obstacles (red)}
    \label{fig:simple}
\end{figure}

The standard MICP encoding (\ref{eq:standard_encoding}) introduces 1600 binary variables and finds a minimum-cost%
\footnote{Diagonal transitions have the same cost as horizontal/vertical transitions.}
satisfying path in 12.7s. The SAT-based approach is slightly faster, finding a solution in 12.3s. Our proposed encoding (\ref{eq:our_encoding}) introduces 11760 binary variables, but takes only 0.47s to find an optimal solution. The relaxation gap is 0.9999 for the standard approach and 0.94 for our proposed method. This supports the idea that even a modest reduction in the relaxation gap can have a significant impact on MICP performance in practice. 

\begin{remark}
    \hl{Note that this simple reach-avoid specification could also be solved as an SPP with LP. This raises the prospect that there may exist yet stronger MICP encodings that would reduce to LP for specifications like this one.}
\end{remark}

We now consider a standard scenario for synthesis from timed temporal logic: in addition to reaching a green goal and avoiding red obstacles, the robot must also visit one of two blue targets (Fig.~\ref{fig:either_or}). The specification is given by
\begin{equation*}
    \eventually{0}{T} (target\_one \lor target\_two) \land \eventually{0}{T} goal \land \always{0}{T} \lnot obstacle
\end{equation*}
with $T=15$. 

\begin{figure}
\end{figure}

SAT and standard MICP both take 25.9s to find a solution, while our proposed approach takes only 0.61s. The relaxation gap is 0.9996 for standard MICP and 0.8739 for our approach. 

We also consider the simpler multi-target scenario shown in Fig.~\ref{fig:convex_relaxation}, where the specification is given by 
\begin{equation*}
    \eventually{T}{T} (target\_one \lor target\_two),
\end{equation*}
and $T=5$. Note that this specification belongs to the fragment (\ref{eq:convex_fragment}) where our proposed encoding is convex, but not the convex fragment induced by the standard MICP encoding. 

\begin{figure}
    \begin{subfigure}{0.49\linewidth}
        \centering
        \includegraphics[width=0.7\linewidth]{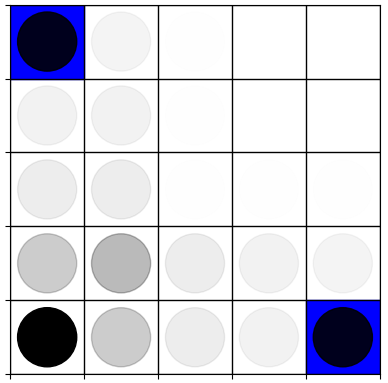}
        \caption{Standard MICP convex relaxation}
        \label{fig:standard_convex_relaxation}
    \end{subfigure}
    \begin{subfigure}{0.48\linewidth}
        \centering
        \includegraphics[width=0.7\linewidth]{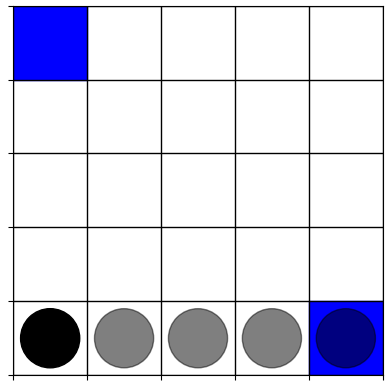}
        \caption{Proposed MICP convex relaxation}
        \label{fig:our_convex_relaxation}
    \end{subfigure}
    \caption{Convex relaxations of the standard encoding (\ref{eq:standard_encoding}) and our proposed encoding (\ref{eq:our_encoding}) for a simple multi-target specification. The relaxation for our approach satisfies the specification directly, since the specification belongs to the convex fragment (\ref{eq:convex_fragment}), but this is not the case for the standard encoding. }
    \label{fig:convex_relaxation}
\end{figure}

Since this problem is relatively small and simple, both MICP methods find the optimal solution rapidly (under 0.01s for both methods). What is more interesting is to consider the solutions to the convex relaxations of each encoding. These are shown in Fig.~\ref{fig:convex_relaxation}, where grid cells are shaded according to the net flow $\sum_{t=0}^T b_s(t)$. The relaxation of our proposed encoding (\ref{fig:our_convex_relaxation}) provides a satisfying solution directly, since $\varphi$ belongs to the fragment \ref{eq:convex_fragment}. The relaxation of the standard encoding (\ref{fig:standard_convex_relaxation}), however, does not respect the transition constraints. Instead, this solution requires the robot to occupy multiple cells at once. 

\hl{The fact that our proposed encoding has a tighter convex relaxation suggests that our approach scale well to long time horizons and complex specifications. We verify this experimentally} by considering a class of more complex, randomly generated scenarios. These scenarios consists of obstacles as well as several groups of targets. The robot is tasked with visiting at least one target in each group while avoiding obstacles. One example of such a specification is shown in Fig.~\ref{fig:multitarget}.

We denote the number of target groups as $N_g$, the number of targets in each group as $N_t$, and the number of obstacles as $N_o$. The specification is given by
\begin{equation}\label{eq:scalable_specification}
    \always{0}{T} \lnot obstacle \land \bigwedge_{k=1}^{N_g} \left( \eventually{0}{T} \bigvee_{l=1}^{N_t} target_k^l \right),
\end{equation}
where $target_k^l$ denotes the $l^{th}$ target in group $k$.

Randomly generating scenarios with a given number of obstacles and targets allows us to test scalability with respect to specification complexity. Specifically, we use the number of target groups, $N_g$, as a proxy for specification complexity. As $N_g$ increases, it becomes more difficult for the robot to find the shortest path that visits each group. 

Specifically, we set up an experiment where scenarios are randomly generated with the following parameters:
\begin{equation*}
    N_t = 2, \quad
    N_o = 2 N_g, \quad
    N = 5 + N_g, \quad
    T = 15.
\end{equation*}
10 trials with the above parameters were considered for several values of $N_g$. We compared the resulting solve times of our proposed MICP encoding with the standard MICP method and SAT-based synthesis. The results are shown in Fig. \ref{fig:scalability_Ng}. The box plot for each trial shows the median (horizontal line), upper and lower quartiles (shaded box) and range (whiskers). 

\begin{figure}
    \centering
    \begin{subfigure}{0.9\linewidth}
        \includegraphics[width=\linewidth]{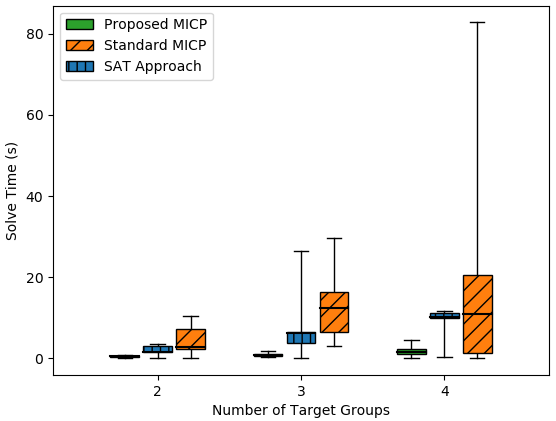}
        \caption{Complexity}
        \label{fig:scalability_Ng}
    \end{subfigure}
    \begin{subfigure}{0.9\linewidth}
        \centering
        \includegraphics[width=\linewidth]{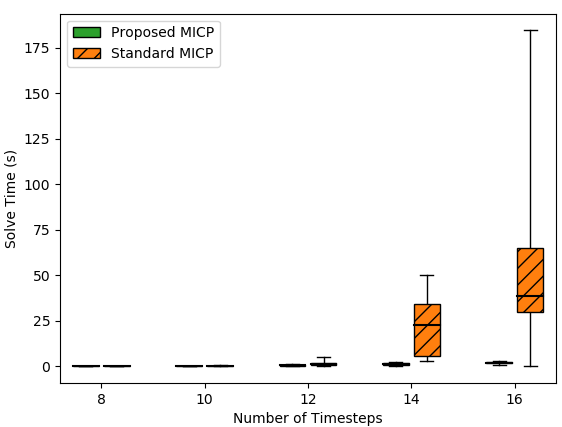}
        \caption{Specification length}
        \label{fig:scalability_T}
    \end{subfigure}
    \caption{Scalability tests with respect to specification complexity and length. Our proposed approach (solid green) consistently outperformed a SAT approach (striped blue) and a standard MICP encoding (striped orange).}
    \label{fig:scalability}
\end{figure}

The standard MICP approach (rightmost bars, diagonally striped orange boxes) performs the worst, with solve times exceeding one minute for the most complex scenarios. The SAT approach (middle bars, blue boxes with vertical stripes) consistently outperforms the standard MICP, which makes sense given the fact that the SAT method finds any feasible solution, while the MICP approach searches for a globally optimal one. Our proposed MICP encoding (left bars, solid greed boxes) outperforms both of the other methods, with all solve times under 5s even for the most complex scenarios. 

It may seem surprising that our proposed approach outperforms the SAT method while also finding a globally optimal solution. We believe that this superior performance is due, again, to the efficiency of branch-and-bound on our proposed encoding. Specifically, the MTL constraints (\ref{eq:our_MICP_occupancy}-\ref{eq:our_satisfaction_constraint}) are all linear in the decision variables, meaning the MTL specification is in some sense always satistisfied (though transition constraints may not be) even for a convex relaxation. This allows the branch-and-bound algorithm to ``hone in'' rapidly on satisfying solutions. SAT solvers, on the other hand, do not have access to this sort of efficient heuristic. 

Finally, we consider scalability with respect to specification length. This is known to be a significant limitation for existing MICP encodings \cite{sadraddini2015robust}, since the number of binary variables increases linearly with the time bound $T$. We consider the same class of randomly-generated scenarios, this time with the following parameters:
\begin{equation*}
    N = 10, \quad
    N_g = 3, \quad
    N_t = 2, \quad
    N_o = 3,
\end{equation*}
and various values of T. The results are shown in Fig.~\ref{fig:scalability_T}. Our approach again consistently outperforms the standard MICP encoding, especially for long time horizons. 

\section{Conclusion}\label{sec:conclusion}

We propose a new MICP encoding for finding an optimal path through a finite-state transition system subject to MTL specifications. By virtue of having a tighter convex relaxation, our proposed approach outperforms existing MICP and SAT-based synthesis methods in terms of speed and scalability to long and complex specifications. Furthermore, this encoding allows specifications within a larger convex fragment to be solved using convex programming directly. Future work will focus on extensions to unbounded specifications, STL (where the advantages of MICP include applicability to systems with high-dimensional PWA physical dynamics), and probabilistic systems. 

\bibliographystyle{IEEEtran}
\bibliography{references}

\end{document}